\documentclass[12pt]{article}

\usepackage{fullpage}

\usepackage{paralist}
\usepackage{cite}
\usepackage{algorithm}
\usepackage{algpseudocode}
\usepackage{amsmath}
\usepackage{amsthm}
\usepackage{color}
\usepackage{amssymb}
\usepackage{subfigure}
\usepackage{times}
\usepackage{epsfig}
\usepackage{graphics}
\usepackage{float}
\usepackage{color}

\usepackage{enumitem}

\usepackage{url}
\usepackage{paralist}
\usepackage{todonotes}

\usepackage{tikz}
\usetikzlibrary{plotmarks,shapes,snakes}
\usepackage{pgfplots}

\usepackage{soul}

\newcommand{\abs}[1]{\left\vert#1\right\vert}
\newcommand{\set}[1]{\left\{#1\right\}}

\newcommand{\ignore}[1]{}

\DeclareMathOperator{\mathalg}{ALG}

\DeclareMathOperator{\OPT}{OPT}

\DeclareMathOperator{\PO}{PO}
\DeclareMathOperator{\NPO}{NPO}

\newtheorem{theorem}{Theorem}
\newtheorem{lemma}[theorem]{Lemma}

\newtheorem{corollary}[theorem]{Corollary}

\newcommand{\POValue}{$MEP$}
\newcommand{\POWork}{$SRPT$}
\newcommand{\POLen}{$LP$}
\newcommand{\fOPT}{$fOPT$}

\definecolor{lightblue}{RGB}{200,200,255}
\definecolor{lightgreen}{RGB}{200,255,200}
\definecolor{darkgreen}{RGB}{100,155,100}

\def\mywid{8cm}
\def\myhei{5.5cm}
\def\legendsize{\footnotesize}

\newcommand{\simullegend}[0]{
\begin{tikzpicture}[scale=0.75]
\node[text width=2cm] (opt) at (1,1) {{\legendsize \ref{gr:FOPTUW} \fOPT}};
\node[text width=2cm] (opt) at (4,1) {{\legendsize \ref{gr:POValue} \POValue}};
\node[text width=2cm] (opt) at (7,1) {{\legendsize \ref{gr:POWork} \POWork}};
\node[text width=2cm] (opt) at (10,1) {{\legendsize \ref{gr:POLen} \POLen}};
\end{tikzpicture}
}

\newcommand{\variablelambda}[5]{ 
\begin{tikzpicture}[scale=0.68]
\begin{axis}[xlabel={$\lambda_{\mathrm{on}}$, $k=10$, $L=#1$, $B=50$},ylabel={#5},ymin=#3,ymax=#3,height=\myhei,width=\mywid, scaled x ticks=true,
x tick label style={font=\footnotesize}, y tick label style={font=\footnotesize} ]
\input{graphs/lambda.k10.B50.L#1$\lambda_{\mathrm{on}}$, $k=10$, $L=#1$, $B=50$.tex}
\end{axis}
\end{tikzpicture}
 }
\newcommand{\variablek}[5]{ 
\begin{tikzpicture}[scale=0.68]
\begin{axis}[xlabel={$k$, $B=50$, $\lambda_{\mathrm{on}}=0.05$, $L=#1$},ylabel={#5},ymin=#3,ymax=#3,height=\myhei,width=\mywid, scaled x ticks=true,
x tick label style={font=\footnotesize}, y tick label style={font=\footnotesize} ]
\input{graphs/k.B50.l005.L#1$k$, $B=50$, $\lambda_{\mathrm{on}}=0.05$, $L=#1$.tex}
\end{axis}
\end{tikzpicture}
 }
\newcommand{\variableb}[5]{ 
\begin{tikzpicture}[scale=0.68]
\begin{axis}[xlabel={$B$, $k=10$, $\lambda_{\mathrm{on}}=0.05$, $L=#1$},ylabel={#5},ymin=#3,ymax=#3,height=\myhei,width=\mywid, scaled x ticks=true,
x tick label style={font=\footnotesize}, y tick label style={font=\footnotesize} ]
\input{graphs/B.k10.l005.L#1$B$, $k=10$, $\lambda_{\mathrm{on}}=0.05$, $L=#1$.tex}
\end{axis}
\end{tikzpicture}
 }

\def\alittle{0cm}

\begin{document}

\title{
Balancing Work and Size with Bounded Buffers
}


\author{Kirill Kogan and Alejandro L\'opez-Ortiz\\
   School of Computer Science \\
   University of Waterloo\\
   \{kkogan,alopez-o\}@uwaterloo.ca \\
   \and Sergey I. Nikolenko \\
   National Research University Higher School of Economics, \\
   Steklov Mathematical Institute, St.~Petersburg, Russia \\
   sergey@logic.pdmi.ras.ru \\
   \and Gabriel Scalosub  and Michael Segal\thanks{The work by Michael Segal has been partly supported by France Telecom, General Motors, European project FLAVIA and Israeli Ministry of Industry, Trade and Labor (consortium CORNET).}\\
   Dept. of Comm. Syst. Eng.\\
   Ben-Gurion University of the Negev\\
   \{sgabriel,segal\}@cse.bgu.ac.il
 }

\maketitle

%

%


\begin{abstract}
We consider a fundamental problem of managing a bounded size queue
buffer where traffic consists of packets of varying size, each
packet requires several rounds of processing before	it can be
transmitted out of the queue, and the goal
is to maximize the throughput, i.e., total size of successfully transmitted packets.
Our work addresses the tension between two conflicting algorithmic approaches:
favouring packets with fewer processing requirements and preferring packets of larger size.
We present a novel model for studying such systems and study the performance of online algorithms that aim to maximize throughput.
\end{abstract}

\section{Introduction}
\label{sec_introduction}

Over the recent years, there has been a growing interest in understanding the effects that buffer sizing has on network performance. The main motivation for these studies is to understand the interplay between buffer size, throughput, and queueing delay.
Broadly speaking, one can identify three main types of delay that contribute to packet latency: transmission and propagation delay, processing delay, and queueing delay.
Recent research that advocates the usage of small buffers in core routers, aiming to reduce queueing delay in the presence of (mostly) TCP traffic, sidesteps the issue that as buffers get smaller, the effect of processing delay becomes much more pronounced~\cite{ramaswamy09analysis}.
The importance of this phenomena is further emphasized by the increasingly processing-heterogeneity of network traffic.
The modern network edge is required to perform tasks with ever-increasing complexity including features such as advanced VPNs services, deep packet inspection, firewall, intrusion detection etc. Each of these features may require a different processing effort at the routers~\cite{wolf00commbench}. Application of such features directly affects processing delay. As a result, the processing order of packets and the way how these packets are processed (``run-for-completion'', processing with preemptions, or some other order) may have a significant impact on the queueing delay and throughput; increasing the required processing per packet in some of the flows may cause increased congestion even for traffic with relatively modest burstiness characteristics.
We should note that in the general case, processing requirements are independent of packet lengths, thus decoupling the amount of work required for a router to process a packet from the throughput gained upon its successful transmission.
Furthermore, it is common for required processing characteristics on a network processor to be highly regular and predictable for a fixed configuration of network elements~\cite{WP+03}. This implies that the per-packet processing requirements are expected to be available and well-defined as a function of the features associated with the flow and the network element configuration.
Moreover, in reactive mode of configuration in Software-Defined Networks (like OpenFlow \cite{OF}), required processing can be estimated by the controller.

This situation leads to several questions that are relevant to the design and implementation of router architectures.
For instance, in light of heterogeneous processing requirements in the traffic, does one need to implement input buffering before a packet is handled by the network processor? If so, what should the size of such a buffer be, and what admission control policy should be applied?
Another question is related to adapting common active queue management (AQM) policies so that they account for heterogeneous processing required by traffic. In this respect, the main question is whether current AQM approaches are capable of considering these characteristics; if not, what form should new policies take?
In this work, we initiate the study of these questions and the tradeoffs they encompass. We focus on improving our understanding of effects that processing disciplines have on throughput in cases of bounded buffers where traffic is heterogeneous in terms of both packet processing requirements and packet length.

In what follows, we adopt the terminology used to describe queue
management within a router in a packet-switched network. We focus
our attention on a general model for the problem where we are
required to manage the admission control and scheduling units in a
single bounded size queue, where arriving traffic consists of {\em
packets}, such that each packet is labeled with its {\em size} (e.g., in
bytes), and {\em processing requirement} (in processor cycles). A
packet is successfully {\em transmitted} once the scheduling unit
has scheduled the packet for processing for at least its required
number of cycles, while the packet resides in the buffer. If a
packet is dropped from the buffer, either upon arrival due to
admission control policies or after being admitted and possibly
partially, but not fully, processed (in scenarios where
push-out is allowed), such a packet is irrevocably lost. We focus
our attention on maximizing the {\em throughput} of the queue,
measured by the total number of bytes of packets that are
successfully transmitted by the queue.

\section {Our Contributions}
In this work we provide a formal model for studying problems of online buffer management and online scheduling in settings where packets have both varying size and heterogeneous processing requirements, and one has a limited size buffer to store arriving packets. Our model lets us study the interplay between potentially conflicting approaches, favoring large packets and favoring packets with less required processing,
in the situation where the goal is to maximize the total throughput. The offline version of this problem is NP-hard, as it encompasses Knapsack as a special case. For the more natural online setting, we provide algorithms with provable performance guarantees as well as lower bounds on the performance of such algorithms.
We focus our attention on priority-based buffer management and scheduling in both {\em push-out} (PO) settings, where admitted packets are allowed to be pushed out of the queue prior to having its processing completed (in which case the packet does not contribute to the system's throughput), and in the {\em non-push-out} (NPO) case, where buffer management decisions are limited to admission control.
Specifically, we consider the following priority queueing regimes:
\begin{inparaenum}[(i)]
\item Shortest-Remaining-Processing-Time (SRPT) first, common in job scheduling environments;
\item Longest-Packet (LP) first;
\item Most-Effective-Packet (MEP) first, which prioritizes packets by the ratio of their residual processing requirement to size.
\end{inparaenum}
We study buffer management algorithms for these priorities, where our bounds are in terms of 
\begin{inparaenum}[(i)]
\item the maximum size of a packet and
\item the maximum number of processing cycles required by a packet.
\end{inparaenum}
A summary of our results appears in Table~\ref{tbl}.

\begin{table*}[t]\centering
\small
\begin{tabular}{|l|c|c|c|c|}\hline
Bound & NPO, any priority & PO, SRPT priority & PO, LP priority & PO, MEP priority \\\hline
Lower & $kL$ & $L$ & $k$ & $1+\left(L\ln k - k - L\right) / B$ \\
Upper & $k(L+1)$ & $2L$ & $(k+3)$  & $-$ \\\hline
\end{tabular}

\caption{Results summary: lower and upper bounds on the competitive ratio.}\label{tbl}
\end{table*}

\section{Related Work}\label{related_work}

In recent years there has been a surge in the study of the effect of buffer size
on the traffic queueing delay arising in the system. Appenzeller et al.~\cite{AKM+04} studied
this problem in the context of statistical multiplexing, focusing mostly on TCP flows. More recently, broader aspects of these question were studied. A comprehensive overview of perspectives on router buffer sizing can be found in~\cite{vishwanath09perspectives}.
Keslassy et al.~\cite{KKSS+11} were the first to consider buffer management and scheduling in the context of network processors, where arriving traffic has heterogeneous processing requirements. They study both FIFO with recycles and SRPT priority schedulers in both push-out and non-push-out buffer management regimes. They focused on the case where packets are of unit size and showed competitive algorithms,
as well as lower bounds, for such settings. They further introduced
the notion of push-out costs which serves to balance the
aggressiveness demonstrated by the the buffer management unit. We
believe the assumption made in~\cite{KKSS+11} that packets are of
unit size is rather restrictive, since in real life NPs have to deal
with packets of varying size, and it is unclear how one should
design algorithms that ensure good throughput guarantees in such
highly heterogeneous scenarios.
The work of Keslassy et al. as well as our current work, can be viewed as part of a larger research effort that focuses on
studying competitive algorithms for buffer management and
scheduling, and specifically the study of such algorithms in
bounded-buffers settings (see, e.g., a recent survey by
Goldwasser~\cite{G+10} which provides an excellent overview of this
field).

\begin{figure}
\begin{center}
\includegraphics{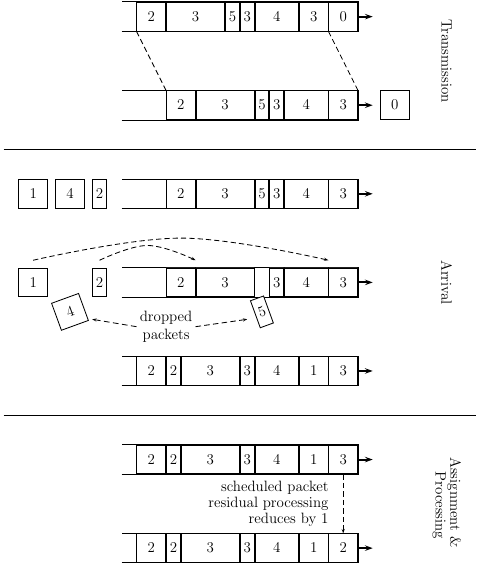}
\end{center}
\caption{An outline of the model. The top subfigure shows the transmission phase, the middle subfigure shows the arrival phase where packets might be discarded, and the bottom subfigure shows the assignment and processing phase. The length of a packet represents its size, and the number stamped on the packet represents the number of its (residual) required processing cycles.}
\label{fig:model}
\end{figure}

The SRPT algorithm has been studied extensively in OS scheduling for multithreaded processors,
and it is well known to be optimal for mean response~\cite{scharge68proof}.
Additional objectives, models, and algorithms have been studied extensively in this context; see,
e.g.,~\cite{LR+97,MRSG+05,MPT+94}.
A comprehensive overview of competitive online scheduling for server
systems can be found in~\cite{K+07}; however, OS scheduling is mostly concerned with average response
time and average slowdown, while we focus on providing
worst-case guarantees on the throughput. Furthermore, OS scheduling does not allow for dropping jobs, which is an inherent aspect of our model, as implied by the fact we have a limited-size buffer, and overflowing packets must be dropped.
The model considered in our work is also closely related to Job-shop scheduling problems~\cite{brucker06jobshop}, most notably to hybrid flow-shop scheduling~\cite{ruiz10hybrid}, in scenarios where machines have bounded buffers. However, while these works focus on system delay, our main focus is system throughput.

\section{Model Description and Algorithmic Framework}

Consider a buffer with bounded capacity of $B$ bytes handling the arrival of a sequence of packets. Each arriving packet $p$ has
a size $\ell(p) \in \set{1,\ldots,L}$ (in bytes) and a number of required processing cycles $r(p) \in \set{1,\ldots,k}$; both
$\ell(p)$ and $r(p)$ are known for every arriving $p$.\footnote{For a motivation why this information may be available, see~\cite{WP+03}.
Our assumption that the size may be as small as one byte is made for simplicity and can be viewed as a scaling assumption.} The
values of maximal required processing $k$ and maximal size $L$ will play a fundamental role in our analysis;
however, that none of our algorithms need to know $k$ in advance. The queue performs two main tasks: {\em buffer
management}, i.e., admission control of new packets and push-out of currently stored packets, and {\em
scheduling}, i.e., which of the currently stored packets are scheduled for processing. The scheduler will be determined
by the {\em priority policy} employed by the queue. We assume a multi-core environment with $C$ processors, so that
at most $C$ packets may be assigned for processing in any given time. Below we assume $C=1$;
this setting suffices to show both the intrinsic difficulties of the model and our algorithmic
scheme. We assume slotted time, where each time slot $t$ consists of 3 phases:
\begin{inparaenum}[(i)]
\item {\em transmission}, when packets with zero remaining required processing leave the queue;
\item {\em arrival}, when new packets arrive, and the buffer management unit performs both admission-control and possibly push-out;
\item {\em assignment and processing}, when a single packet is assigned for processing by the scheduling unit.
\end{inparaenum}
Figure~\ref{fig:model} depicts our general model.
If a packet is dropped prior to being {\em transmitted} (i.e., while it still has a positive number of required processing cycles), it is lost;
we can drop a packet either upon arrival or due to a push-out decision while it is in the buffer. A packet contributes its size to the
objective function only upon being successfully transmitted. The goal of a buffer management algorithm is to maximize the overall throughput,
 i.e., total number of bytes transmitted.

We define a {\em greedy} buffer management policy as a
policy that accepts all arrivals whenever there is available
buffer space in the queue. We only consider {\em work-conserving} schedulers, i.e. schedulers that never leave the processor idle unnecessarily.
An arriving packet $p$ {\em pushes out} a packet $q$ that has already been accepted into the buffer
iff $q$ is dropped in order to free up buffer space for $p$ and $p$ is admitted to the buffer instead. A buffer management policy is called
{\em push-out} if it allows packets to push out currently stored packets.
For an algorithm $\mathalg$ and a time slot $t$, we define
$IB^{\mathalg}_t$ as the set of packets stored in $\mathalg$'s buffer  at time $t$.
The number of {\em processing cycles} of a packet is key to
our algorithms. Formally, for every time $t$, and
every packet $p$ currently stored in the queue, its number of {\em residual
processing cycles}, denoted $r_t(p)$, is defined to be the number of
processing cycles it requires before it can be successfully
transmitted.

Push-out and non-push-out policies will be denoted by $\PO$ and $\NPO$ respectively.
We will focus our attention on {\em priority-queueing} disciplines, which determine both the scheduling and the buffer management behaviour of the queue. Specifically, we will focus our attention on the following priorities, which differ by the parameter determining the priority:
\begin{inparaenum}[(i)]
\item {\em processing}, in which the packet with the least amount of residual cycles has the top priority (referred to as SRPT);
\item {\em length}, in which the largest packet receives the top priority (referred to as LP).
\item {\em processing-to-length}, in which the packet with the least residual cycles to size ratio has the top priority (referred to as MEP).
\end{inparaenum}

\begin{algorithm}{}
\normalsize
\begin{algorithmic}[1]
   \If {there is space available in the queue}
     \State accept $p$
   \EndIf
\end{algorithmic}
\caption{{\sc $\NPO$}($p$): Buffer Management Policy}
\label{alg:npo}
\end{algorithm}
\begin{algorithm}{}
\normalsize
\begin{algorithmic}[1]
   \State accept $p$
   \While {the last packet $q$ in the buffer starts above  position $B-2L+1$}
     \State drop $q$
   \EndWhile
\end{algorithmic}
\caption{{\sc $\PO$}($p$): Buffer Management Policy}
\label{alg:pq-1}
\end{algorithm}

We use competitive analysis~\cite{ST+85,Borodin-ElYaniv} to evaluate performance guarantees provided by our online algorithms. An
algorithm ALG is said to be {\em $\alpha$-competitive} (for some $\alpha \geq 1$) if for any arrival sequence $\sigma$, the total length of packets successfully transmitted by ALG is at least $1/\alpha$ times the total length of packets successfully delivered by an optimal solution (denoted $\OPT$), obtained by an offline clairvoyant algorithm.

Next we define the algorithms used below for all types of characteristics.
The Non-Push-Out Algorithm (NPO) is a simple greedy work-conserving policy that accepts a packet if there is buffer space available.
In the push-out case, the $\PO$ algorithm is defined in Algorithm~\ref{alg:pq-1}. Note that $\PO$ is somewhat conservative in its use of the buffer; the reason for this will be clear from our results presented in Sections~\ref{sec:pq-residual-passes} and~\ref{sec:pq-residual-passes-byte}.
We will sometimes use the term {\em value} to denote the total length of a set of packets, and our analysis will be based on
comparing the mapping value obtained by an optimal solution to that of our algorithm. Specifically, we will make use of
mappings between packets transmitted by $\OPT$ and by our algorithm such that their respective values differ only by a
multiplicative factor; this factor will serve as a bound on the competitive ratio.

\section{Useful Properties of Ordered Multisets}

To facilitate our proofs, we will make use of properties of ordered (multi)sets. These properties enable us to compare the performance of our proposed algorithms with the optimal policy for various priority disciplines.
In what follows, we consider multisets of real numbers, where we assume that each multiset is ordered in non-decreasing order. We will refer to such multisets as {\em ordered sets}. For every $1 \leq i \leq \abs{A}$, we will further refer to element $a_i \in A$ or $A[i]$ as the $i$-th element in the set $A$ in the above-mentioned order.
Given two ordered sets $A,B$, we say $A \geq B$ if $a_i \geq b_i$ for every $i$ such that both $a_i$ and $b_i$ exist.
The following lemma and its corollary will be needed as fundamental tools throughout our analysis.

\begin{lemma}
\label{l:add_pair}
For any two ordered sets $A,B$ satisfying $A \geq B$ and any two real numbers $a,b$ such that $a \geq b$, if \begin{inparaenum}[(i)]
\item \label{>=} $b\leq b_{|B|}$ or \item \label{<=} $|A|\leq |B|$ \end{inparaenum} then ordered sets $A'=A\cup\set{a}$, $B'=B\cup\set{b}$ satisfy $A' \geq B'$.
\end{lemma}
\begin{proof}

\begin{figure}[t!]
\centerline{
\hfill
\subfigure[Case $i \leq j$]{
\includegraphics[scale=0.8]{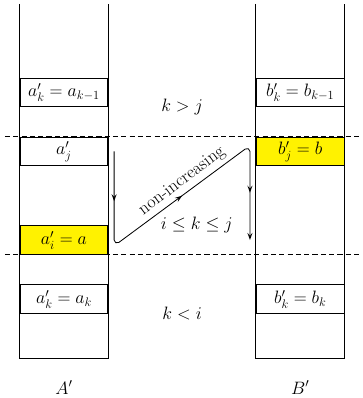}
}
\hfill
\subfigure[Case  $i > j$]{
\includegraphics[scale=0.8]{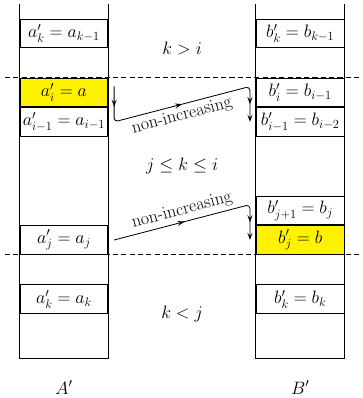}
}
\hfill
}
\caption{Cases of Lemma~\ref{l:add_pair}.} \label{fig:ordered-set}
\end{figure}

We will refer to elements in $A'$ and $B'$ as $a'$ and $b'$, respectively.
Assume $i$ and $j$ are the positions of $a \in A'$ and $b \in B'$, respectively, i.e., $a'_i=a$ and $b'_j=b$. We need to show that for every $k$ for which both $a'_k$ and $b'_k$ exist, $a'_k \geq b'_k$.
We distinguish two cases:
\begin{enumerate}
\item [(a)] $i \leq j$ (see Figure~\ref{fig:ordered-set}(a)):
for all $k < i$, $a'_k=a_k$, and $b'_k=b_k$, hence by the assumption that $A \geq B$, $a'_k \geq b'_k$.
By the assumption that $a \geq b$, and the fact $A'$ and $B'$ are ordered, for every $p \geq i$ and $q < j$ we have $a'_p \geq a'_i \geq b'_j \geq b'_q$.
In particular, for every $i \leq k < j$ we have $a'_k \geq b'_k$ (by taking $p=q=k$).
For $k=j$, since $A'$ and $B'$ are ordered, and since in the current case $i \leq j$, we have $a'_j \geq a'_i = a \geq b = b'_j$.
For $k > j$ we have $a'_k=a_{k-1} \geq b_{k-1} = b'_k$, where the inequality follows from the assumption that $A \geq B$.
\item [(b)] $i > j$ (see Figure~\ref{fig:ordered-set}(b)):
for all $k < j$, $a'_k=a_k$, and $b'_k=b_k$, hence by the assumption that $A \geq B$, $a'_k \geq b'_k$.
For $k=j$, $b'_k \leq b'_{k+1} = b_k \leq a_k = a'_k$, which follows from the fact that $b$ is inserted in slot $j=k$, $B'$ is ordered, the assumption that $A \geq B$ and  $b\leq b_{|B|}$ or $|A|\leq |B|$.
For $j < k < i$, $b'_k = b_{k-1} \leq a_{k-1} \leq a_k = a'_k$, which follows from the assumption that $A \geq B$.
For $k = i$, $a=a'_i\geq a_{i-1}\geq b_{i-1}=b'_i$.
For $k > i$, $a'_k = a_{k-1} \geq b_{k-1} = b'_k$.
\end{enumerate}
We are therefore guaranteed to have $A' \geq B'$, as required.
\end{proof}

\begin{corollary}
\label{c:add_single}
For any two ordered sets $A,B$ satisfying $A \geq B$, and any real number $b$, if \begin{inparaenum}[(i)]
\item \label{>=2} $b\leq b_{|B|}$ or \item \label{<=2} $|A|\leq |B|$ \end{inparaenum} then
the ordered set $B'=B\cup\set{b}$ satisfies $A \geq B'$.
\end{corollary}
\begin{proof}
Assume $b$ is inserted in $B'$ in location $j$. Consider a virtual item $a$ such that $a>\max\set{a_{\abs{A}},b}$. We now consider adding both $a$ and $b$ to sets $A$ and $B$ respectively. By Lemma~\ref{l:add_pair}, it follows that the resulting sets $A',B'$ satisfy $A'\geq B'$. Notice that the first $\abs{A}$ elements of $A'$ are exactly the set $A$ (by the choice of $a$), implying that we also have $A \geq B'$.
\end{proof}

\section{Non-Push-out Policies}\label{sec:non-preemptive}
While the $\NPO$ algorithm may have different priorities that govern its admission policy (which packets to choose from a set of
simultaneously arriving packets), it cannot push already admitted packets out. As a result, the worst-case bounds are very similar for all
three priorities we consider, and we simply prove a unified lower and upper bound on the performance of $\NPO$ independent of 
admission policy.

\begin{theorem}\label{t:srpt-npo-lower-upper-bound}
$\NPO$ is at least $kL$-competitive and at most $k(L+1)$-competitive.
\end{theorem}
\begin{proof}
We begin with the lower bound. To show a lower bound, we need to present a ``hard'' sequence of arriving packets.
Consider a burst of $B$ $1$-byte packets with $k$ processing cycles arriving on the first time slot; $\NPO$ invariably accepts
them all and begins processing, while $\OPT$ is free to reject them. On the second time slot, there arrive $B$ $L$-byte packets
with $1$ processing cycle each; they are accepted by $\OPT$, and it begins processing.
After that, every $k$-th time slot there arrives a $1$-byte packet with $k$ processing cycle (to fill up $\NPO$ buffer),
and on other time slots $L$-byte packets with $1$ processing cycle arrive, filling up $\OPT$ queue. As a result,
$\OPT$ transmits $L$ bytes per time slot while $\NPO$ transmits $1$ byte per $k$ time slots, getting the bound in question
(asymptotically, since $\NPO$ is working on the very first time slot).

To prove the upper bound, note that $\NPO$ must fill up its buffer before it drops any packets.
Moreover, so long as the $\NPO$ buffer is not empty, after at most $k$ time steps $\NPO$
must transmit its HOL packet. This means that $\NPO$ is transmitting at a
rate of at least $1$ byte per $k$ time steps, while $\OPT$ can transmit at most $k$ packets of size $L$ each over $k$ time slots.
Hence, the number of transmitted bytes at time $t$ for $\NPO$ is at least
$ t/k $ (we assume $k$ divides $t$ evenly for simplicity of exposition)
while $\OPT$ transmitted at most $tL$ bytes for a competitive ratio of $kL$ so long as the NPO buffer is not empty.

If $\NPO$ empties its buffer first, this means that the $\NPO$ buffer was congested at some point, so $\NPO$
has transmitted at least $B$ bytes, and $\OPT$ can transmit at most $B$ more bytes before more packets arrive to
the $\NPO$ queue. The overall ratio is therefore at most $\frac{tL+B}{t/k}$ under the condition that
$t \ge B$, yielding the bound.
\end{proof}

Thus, the simplicity of
non-push-out greedy policies does have its price. In the
following sections we explore the benefits and analyze performance of push-out policies.

\section{Buffer Management with SRPT-based Priorities}
\label{sec:pq-residual-passes}

%

In this section we address the  buffer management problem of when the queueing discipline gives higher
priority to packets with fewer required processing cycles.
We show first a lower and then an upper bound for the $\PO$ Algorithm~\ref{alg:pq-1}
with SRPT priorities. In this and subsequent sections we focus our attention on the push-out case since
non-push-out results have already been shown in Section~\ref{sec:non-preemptive} for all considered priorities.

\subsection{Lower bound}\label{sec:pq1}

\begin{theorem}
\label{t:srpt-po-lower-bound}
For $B > 2L$, $\PO$ is at least $L$-competitive for SRPT-based priorities, .
\end{theorem}
\begin{proof}[Proof of Theorem~\ref{t:srpt-po-lower-bound}]
Assume that $B/L$ is an integer. All packets received will have a
single residual pass. Consider the following sequence of arrivals.
At the beginning $B-2L+1$ 1-byte packets arrive. $\PO$ accepts
all of them. $\OPT$ drops all of them. Later on during the same time slot
$B/L$ packets of length $L$ arrive, each with a single residual
pass. $\PO$ drops all of them since their value is no better than
the value of packets in its buffer, but $\OPT$ accepts all of them and
thus $\OPT$ buffer is full. During each following time slot one 1-byte packet arrives,
each requiring a single processing cycle, followed by one packet of size $L$ bytes, requiring a single processing cycle.
$\PO$ accepts all 1-byte packets
but it does not accept any of the $L$-bytes packets. Thus, for each time
slot when there are arrivals, $\OPT$ transmits a packet of size $L$, and at the same time $\PO$ transmits a 1-byte packet.
At the end, $\OPT$ transmits $B$ bytes while $\PO$ transmits $B-2L+1$ additional bytes. Therefore,
$B+nL$ and $B-2L+1+n$ bytes are transmitted by $\OPT$ and $\PO$, respectively, where $n$ is a number of time slots with non-empty arrivals. We obtain that for $n>>B$, $\PO$ cannot have a competitive ratio better than $L$.
\end{proof}

\subsection{$\PO$ Upper Bound for $B>2L$}

Next, we show one of our main results, an upper bound for $\PO$ with SRPT priorities.

\begin{theorem}
\label{t:srpt-po-upper-bound}
For $B > 2L$, $\PO$ is at most $4L-2$-competitive for SRPT-based priorities.
\end{theorem}

In what follows we assume that $\OPT$ never pushes out packets. Such an optimal solution exists since one can consider the whole
input being available to $\OPT$ a priori. Thus, all packets accepted
by $\OPT$ are transmitted. Our analysis will be based on describing a mapping of packets in $\OPT$'s buffer to packets transmitted by $\PO$, such that every packet $q$
transmitted by $\PO$ has at most $4L-2$ bytes of $\OPT$ associated
with it. To facilitate the exposition we describe packet processing
as if packets arrive individually and sequentially one at a time,
although in reality more than one packet might arrive at a single time step $t$.
The mapping will be dynamically updated for each packet arrival and
for each packet transmission, in both $\OPT$ and $\PO$.

\paragraph{Mapping routine.}
During the transmission phase we distinguish between three cases:
\begin{itemize}[noitemsep,topsep=0pt,parsep=0pt,partopsep=0pt,listparindent=0pt,itemindent=0pt]
\item [\bf T0]  If both $\OPT$ and $\PO$ do not transmit then the mapping remains unchanged.
\item [\bf T1] If $\PO$ transmits a packet $q$ then we remove its mapped image in $\OPT$'s buffer
from future consideration in the mapping.
 The subset of these $\OPT$ packets or bytes that stay in $\OPT$ buffer at the end of transmission phase are called of type 1.
\item [\bf T2] If $\OPT$ transmits a packet $p$ but its mapped packet $q$ in $\PO$ is not transmitted then $p$ is termed a packet of type 2. (We will show next that this case never occurs).
\end{itemize}

At time $t$, denote by $M^O_t$ the ordered set of residual pass values for all non-type 1 $\OPT$ packets. All $M^O_t$ values are grouped into blocks in the following way. A block is a minimal subset of consecutive $M^O_t$ values starting from the lowest position that is not covered by any previous block, such that the overall length of the packets associated with the block is at least $L$. The minimal value in each block is called a block {\em representative}. Denote by $R_t$ an ordered set of representatives at time $t$. In addition we denote by $M^P_t$ an ordered set of processing cycles values of packets in $\PO$'s buffer at time $t$.

\begin{figure}[t!]
\centerline{
\subfigure[$(P,i)$-mapping-shift]{
\includegraphics[scale=0.8]{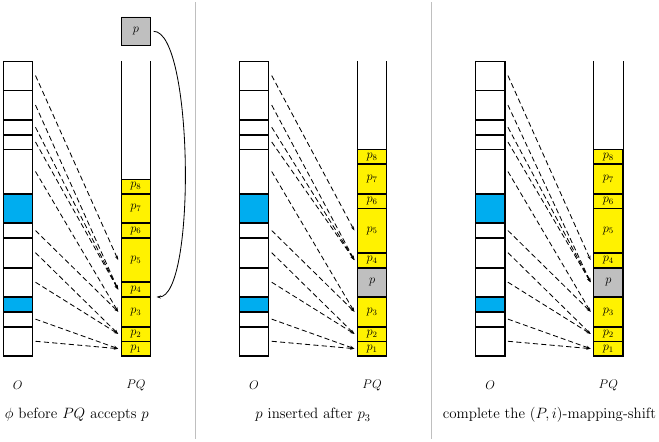}
\label{fig:mapping-shift-alg}
}}
\centerline{
\subfigure[$(O,j)$-mapping-shift ]{
\includegraphics[scale=0.8]{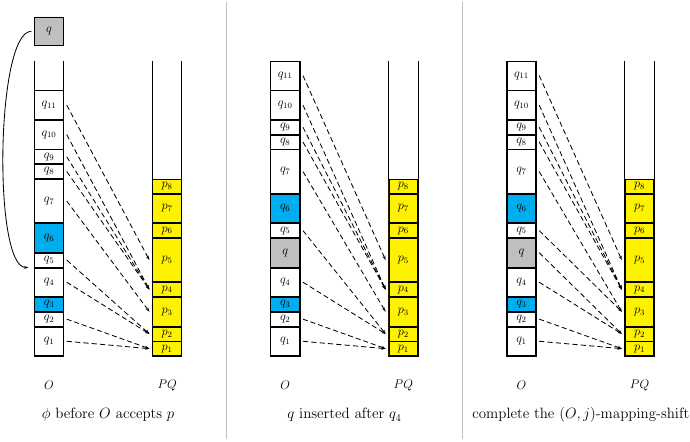}
\label{fig:mapping-shift-opt}
}}
\caption{Example of the mapping used in the proof of Lemma~\ref{l:srpt_po_faster}, and the mapping shifts performed by the analysis. The white $\OPT$ packets should be mapped and white $\PO$ packets are available for mapping. The blue $\OPT$ packets are of type 1.}
\label{fig:mapping-shift}
\end{figure}

After the arrival at time $t$ of a packet $p$ we distinguish between the following cases:

\begin{itemize}[noitemsep,topsep=0pt,parsep=0pt,partopsep=0pt]
\item [\bf A0] If $p$ is not accepted by both $\OPT$ and $\PO$, then the mapping remains unchanged.
\item [\bf A1] If after acceptance of $p$ some $\PO$ packets were dropped then clear the mappings by step A1 between these $\PO$ packets and  its mapped $\OPT$ mates. If $p$ remains in $\PO$'s buffer and $p$ is an $i$-th packet in it perform a {\em $(P,i)$-mapping-shift} (see Figure~\ref{fig:mapping-shift-alg}): for each non-empty $j$-th block $b$ and $j$-th $\PO$ packet $q$, with $j\geq i$ clear the mapping to $q$ by step A1 and map all packets of block $b$ to $q$. If $p$ is accepted by $\OPT$ to the $j$-th block, perform  an {\em $(O,j)$-mapping-shift} (see Figure~\ref{fig:mapping-shift-opt}): clear all mappings by step A1 between packets of the old $l$-th block and $l$-th $\PO$ packet (if both exist), $l\geq j$, recompute blocks starting from the $j$-th and map packets of $l$-th block to $l$-th $\PO$ packet if both exist, $l\geq j$.
\item [\bf A2] Clear all mappings assigned by step A2.  Map packets of all unmapped blocks to the HOL $\PO$ packet.
\end{itemize}

\begin{lemma}\label{lem:feas_and_blocksize}
\begin{enumerate}[label=(\arabic*)]
\item The mapping is feasible.
\item The total length of packets of the same block is at most $2L-1$.
\end{enumerate}
\end{lemma}
\begin{proof}
1. By definition $\PO$ accepts the arriving packet and all the packets
with packet start above $B-2L+1$ are dropped. Hence, if after applying step A1 of the mapping routine there are still
unmapped $\OPT$ packets then the $\PO$ buffer must contain at least one packet. Therefore, all $\OPT$ packets unmapped by step A1
are mapped by step A2.

2. In the worst case an total length of all packets in the block except the last one is $L-1$ and the last packet of the same block has length $L$, so the claim follows.
\end{proof}

\begin{lemma}\label{l:srpt_po_faster}
After the $t$-th packet arrives, if an $\OPT$ packet $p$ is mapped to a (possibly transmitted) $\PO$ packet $q$
then $r_t(p)\geq r_t(q)$. Moreover, all $\OPT$ packets are mapped, and
at most $2L-1$ bytes are mapped to each $\PO$ packet by step A1, and possibly at most $2L-1$ more bytes are mapped
to the HOL packet by step A2 at any time $t$.
\end{lemma}
\begin{proof} 
We prove the lemma by induction on the number of arrived packets.
For the base, consider the first arriving packet $p$; by definition $\PO$ always accepts it.
If $p$ is dropped by $\OPT$ then the claim trivially holds.
If $p$ is accepted by $\OPT$, it creates a new block with representative $p$. Clearly, $r_1(p)\geq r_1(p)$, all $\OPT$
packets are mapped, and at most $L$ bytes are mapped to the $\PO$
packet $p$, and the base holds.

Assume by induction that for any time $t'<t$, after the arrival of the $t'$-th packet it holds that for
any $\OPT$ packet $p$ that is mapped to a (possibly transmitted) $\PO$
packet $q$, $r_{t'}(p)\geq r_{t'}(q)$. Moreover, all $\OPT$ packets
are mapped and at most $2L-1$ $\OPT$ bytes are mapped to each $\PO$ packet  by step A1.
In addition at most $2L-1$ $\OPT$ bytes are mapped to the HOL packet in $\PO$ buffer at time $t'$ by step A2.

Clearly, if a representative of a block $p'$ mapped to a $\PO$ packet $q$ by step A1 satisfies
$r_{t'}(p')\geq r_{t'}(q)$ at time $t'$, then for any packet $p''$ of the same block $r_{t'}(p'')\geq r_{t'}(q)$.
To show that this holds after the $t$-th packet arrives, it suffices to consider the ordered set of representatives
$R_{t-1}$ and its update after this arrival and show that $R_t \geq M^P_t$. By the induction hypothesis
the remaining number of processing cycles of any $\OPT$ packet is at least the number of processing cycles of its $\PO$ counterpart,
i.e., $R_{t-1}\geq M^{P}_{t-1}$, and there are no $\OPT$ packets of type 2 formed while the first $t-1$ packets were accepted.
Denote by $R^1_{t'}$ a set of representatives of blocks mapped by step A1. Since all packets of the same block are
mapped to the same $\PO$ packet, $\abs{R^1_{t-1}}\leq\abs{M^{P}_{t-1}}$.

We denote by $t-$ the time slot just before the arrival of the $t$-th
packet. First suppose that a transmission occurs before the $t$-th packet arrives, i.e., between the
$t$-th and $(t-1)$-th packet arrivals at least one packet is transmitted by $\OPT$ or $\PO$. By the induction hypothesis, it is
impossible for $\OPT$ to transmit a packet corresponding to the first
value in $M^P_{t-1}$ before the packet whose value is the first in $M^P_{t-1}$,
and this holds for any sequence of transmissions prior to the $t$-th arrival. Therefore, if
$\PO$ transmits between the $t-1$-th and $t$-th arrival, $\abs{M^P_{t-}}$ is reduced by one. On the other hand,
$\abs{M^O_{t-}}$ is reduced by the number of packets in the first block (if any) mapped by step A1 to the packet sent by
$\PO$. Hence,  $\abs{R^1_{t-}}\leq\abs{M^{P}_{t-}}$. Moreover, any value mapped by step A2 to a packet transmitted by $\PO$
is removed from $M^O_{t-}$ upon this transmission
(by the definition of $M^O_t$ which consists of non-type 1 packets only). Thus, $R_{t-}=R^1_{t-}$ and $\abs{R_{t-}}\leq\abs{M^{P}_{t-}}$
in this case, and the claim holds at time $t-$, in particular $R_{t-}\geq M^{P}_{t-}$.

Consider now the arrival of the $t$-th packet $p$. We distinguish the following cases.
\begin{itemize}[noitemsep,topsep=0pt,parsep=0pt,partopsep=0pt,listparindent=0pt,itemindent=0pt]
\item [\bf 0]  $\OPT$ does not accept $p$ and $\PO$ accepts and immediately drops $p$. We are done.
\item [\bf 1] $\PO$ does not drop $p$, $\OPT$ does not accept $p$. In this case, $R_t=R_{t-}$ and it suffices to show that $R_{t-}\geq M^P_t$.
\begin{itemize}[noitemsep,topsep=0pt,parsep=0pt,partopsep=0pt,listparindent=0pt,itemindent=0pt]
\item [\bf 1.1] $\abs{R_{t-}}\geq \abs{M^{P}_{t-}}$. Since $\abs{R^1_{t-}}\leq\abs{M^{P}_{t-}}$, some $\OPT$ packets represented in $R_{t-}$ are mapped by step A2.
In this case, the last packet in $\PO$ buffer occupies the $(B-2L+1)$-th byte (each block has length $\geq L$, each block is mapped to a single $\PO$ packet, and all bytes in $\PO$ buffer are available for mapping). $\PO$ does not drop $p$, so the value of $r_t(p)$ is at most the last value in $M^P_t$. Since $\OPT$ does not accept $p$, by Corollary~\ref{c:add_single}(\ref{>=2}), $R_{t-}\geq M^{P}_{t-}\cup \set{r_t(p)}=M^P_t$.
\item [\bf 1.2] $\abs{R^1_{t-}}\leq \abs{M^{P}_{t-}}$. Again, since in this case $\OPT$ does not accept $p$, by Corollary~\ref{c:add_single}(\ref{<=2}), $R_{t-}\geq M^{P}_{t-}\cup \set{r_t(p)}=M^P_t$.
\end{itemize}
\item [\bf 2] $\OPT$ accepts $p$, $\PO$ drops $p$. In this case $M^P_t=M^P_{t-}$, and $r_t(p)$ is larger than any value in $M^P_{t}$. Let $l$ be the position of $p$ in $\OPT$ buffer. For any $m \geq l$ the $m$-th $\OPT$ packet has more residual cycles than any value in $M^P_t$,
so for any $\OPT$ packet $p'$ mapped to $\PO$ packet $q$ by step A1 $r_t(p')\geq r_t(q)$, and we have $R_t \geq M^P_t$.
\item [\bf 3] $\OPT$ accepts $p$, $\PO$ does not drop $p$. If $\abs{R_{t-}}\leq \abs{M^{P}_{t-}}$ or $\abs{R_{t-}}\geq \abs{M^{P}_{t-}}$, then similarly to the Cases 1.1 and 1.2, by Lemma~\ref{l:add_pair} we have that
    $R'=R_{t-}\cup\set{r_t(p)}\geq M^P_{t-}\cup \set{r_t(p)}=M^P_t.$
    Therefore, in this case it suffices to show that $R_t\geq R'$, which in turn implies $R_t\geq M^P_t$.
    Let $j$ denote the index of the block where $p$ is inserted in $\OPT$. We have to consider two possibilities for the position of $p$'s number of processing cycles in $R'$, which could be either the $j$-th or $(j+1)$-st.
\begin{itemize}[noitemsep,topsep=0pt,parsep=0pt,partopsep=0pt,listparindent=0pt,itemindent=0pt]
\item [\bf 3.1] $R'[j]=r_t(p)$. In this case $p$ now serves as the representative of block $j$, i.e., $R_{t}[j]=r_t(p)$.
If $\ell(p)=L$ then $p$ forms a full block and $R_t[m]=R^1_{t-}[m-1]$ for all $m>j$.
Therefore, $R_t\geq R'$ (actually, in this case we have strict equality). Otherwise, $\ell(p)<L$, and at least as many residual processing 
cycles as the $(l+1)$-th element will join the $j$-th block after recomputation (since such a block must add at least $L$ to total length).
We therefore have $R_t[m]\geq R^1_{t-}[m]$ for all $m>j$. Since $R'[m]=R_{t-}[m-1]$, $m>j$, this case follows.
\item [\bf 3.2] $R'[j+1]=r_t(p)$. In this case the representative of block $j$ remains unchanged, i.e., $R'[j] =R_{t}[j]=R_{t-}[j]$ and $R'[m]=R_{t-}[m-1]$, $m>j+1$.
Since $p$ belongs to the $j$-th block after acceptance and $r_t(p)$
is not a representative of the block, then $R_t[j+1]\geq r_t(p)$.
Since after recomputation representatives will move up for no more
than one block in $R_t$ compared to $R_{t-1}$, $R_t[m]\geq
R_{t-}[m-1]$, $m>j+1$.
Therefore, $R_t\geq R'$ and this case follows.
\end{itemize}
\end{itemize}

Now let us show that there are sufficiently many $\PO$ packets
to map all of $\OPT$ packets such that at most $2L-1$ bytes
are assigned by step A1 to each transmitted packet of $\PO$ and
additionally at most $2L-1$ bytes are assigned to the HOL packet of
$\PO$ by step A2. Recall that the claim holds for time $t-$.
Consider the arrival of the $t$-th packet $p$. If $\PO$ accepts
$p$, then the claim holds, since this new packet can support the
block changes (and possible addition) that may potentially occur if
$\OPT$ also accepts $p$. If $\PO$ does not accept $p$ then by the
definition of $\PO$ this can only happen if the buffer occupancy of
$\PO$ is at least $B-2L+1$. Clearly, the total length of a block
mapped by step A1 to any $\PO$ packet is at most $2L-1$ (by
definition). Furthermore, since we have shown that $R_t \geq M^P_t$,
and by definition the blocks are of total length at least $L$, it
must follow that the total length of packets in $\PO$ covers at
least this amount of total length of packets in $\OPT$ mapped to
$\PO$ by step A1. It follows that the remaining total length of
packets in the buffer of $\OPT$ that are not mapped by step A1 can be
at most $2L-1$ (the possibly unused space in $\PO$). It follows
that the total length of packets mapped to the HOL packet of
$\PO$ by step A2 is at most $2L-1$, as required.
\end{proof}

The proof of Theorem~\ref{t:srpt-po-upper-bound} now follows immediately from Lemma~\ref{l:srpt_po_faster}.

Next we generalize the previous mapping and show how to improve the upper bound of $\PO$ for sufficiently large buffers.
\begin{theorem}
\label{t:srpt-po-upper-bound-new}
$\PO$ is at most $\frac{(2L-1)(N+1)}{N}$-competitive for SRPT-based priorities, where
$N=\lceil\frac{B-2L+1}{2L-1}\rceil$.
\end{theorem}

The idea is to redistribute mapped bytes by step A2 between different $\PO$ packets.
Let $N=\lceil\frac{B-2L+1}{2L-1}\rceil$. We consider an updated version of step A2 and now to each $\PO$ packet can be mapped at most $\frac{2L-1}{N}$ value by step A2.
The mapping routine is unchanged during the transmission phase and now it operates on $M^O_t$ in the following way.
Denote by $M^O_t$ at time $t$ an ordered set of values of processing cycles of no-type 1 $\OPT$ packets that are not mapped by the step A2 as defined below. Observe that now we exclude from the future consideration by step A1 all $\OPT$ packets that are mapped by step A2 even before its $\PO$ mates are transmitted. The definition of block, representative, $R_t$ and $M^P_t$ remain unchanged.
During the arrival of a packet $p$ at time $t$, steps A0 and A1
remain unchanged. Next, we define the changed or new steps.

\begin{itemize}[noitemsep,topsep=0pt,parsep=0pt,partopsep=0pt,listparindent=0pt,itemindent=0pt]
\item [\bf A2]
If prior to $t$ there are no bytes mapped by step A2 and after the $t$-th A1 step there are still $Y$ unmapped $\OPT$ bytes then a $j$-th portion of $\frac{Y}{N}$ bytes unmapped by step A1 map to $\PO$ packet whose mapped block contains the $((j-1)(2L-1)+1)$-st byte $x$, $1\leq j\leq N$. We say that $x$ ``defines'' a mapping of this portion of still unmapped bytes.
Let $Y$ be the overall length mapped by step A2 prior to time $t$ and still there are $Y_0$ unmapped bytes after applying step A1 during time $t$. Let the mapping of the $Y$-th byte assigned by step A2 be the $l$-th byte in the $\OPT$ buffer. Map each $j$-th portion of $\frac{Y_0}{N}$ still unmapped by step A1 byte to $\PO$ packet whose mapped block contains the $(j(2L-1)+l+1)$-st byte, $1\leq j\leq N$. Observe that both these bytes can be remapped to the other $\PO$ packet during the \emph{ $(O,j)$-mapping-shift }.
\item [\bf A3] Values unmapped by steps A1 and A2 are assigned to the HOL $\PO$ packet. We will show that step A3 is never applied and 
is required only for completeness.
\end{itemize}

The mapping is feasible since during arrivals the $\PO$ buffer contains at least one packet and any value that is unmapped by Steps A1 and A2 is assigned by step A3 to the
HOL $\PO$ packet. Lemma~\ref{lem:feas_and_blocksize}(2) remains the same.
The next lemma is very similar to Lemma~\ref{l:srpt_po_faster}. Namely, if an $\OPT$ packet $p$ is mapped by step A1 to a (possibly transmitted) $\PO$ packet $q$ then $r_t(p)\geq r_t(q)$. The fact that the total value assigned to each $\PO$ packet is at most $\frac{(2L-1)(N+1)}{N}$ follows from the fact that for each $\OPT$ packet $p$ that is mapped by step A1 to a $\PO$ packet $q$ at any time $t$, $r_t(p)\geq r_t(q)$, the maximal block size is $2L-1$ bytes. Theorem~\ref{t:srpt-po-upper-bound-new} follows immediately from Lemma~\ref{l:srpt-po-faster-new}.

\begin{lemma}
\label{l:srpt-po-faster-new}
After arrival of the $t$-th packet, if an $\OPT$
packet $p$ is mapped to a (possibly transmitted) $\PO$ packet $q$
then $r_t(p)\geq r_t(q)$. Moreover, all $\OPT$ packets are mapped and
at most $\frac{(2L-1)(N+1)}{N}$ value is mapped to each $\PO$ packet at time $t$, where $N=\lceil\frac{B-2L+1}{2L-1}\rceil$.
\end{lemma}

\begin{corollary}
\label{t:srpt-best}
If $B > 4L^2-2L$ then $\PO$ is at most
$2L$-competitive for SRPT-based priorities.
\end{corollary}

\section{Buffer Management with LP-based Priorities}\label{sec:pq_length}

\subsection{Lower bound}\label{sec:pq_length_preemptive}

We begin with a lower bound for the $\PO$ algorithm with LP-based priorities.
and then proceed to an upper bound of $\PO$ with LP-based priorities.
\begin{theorem}
\label{t:lp-po-lower-bound} $\PO$ is at least $k$-competitive for LP-based priorities on a sufficiently long sequence.
\end{theorem}
\begin{proof}
Here, we will consider a push-out version of $\OPT$  for simplicity of
description. Assume $\frac{B}{L}$  be an integer value. Consider a cycle of $L$ iterations of the first type and later sequence of $n>0$  iterations
of the second type (defined below). Each iteration of the first type contains $k-1$ time slots. At the beginning of the $i$-th iteration of the first type
$\lceil\frac{B}{i}\rceil$ packets of $i$  bytes with $k$ processing cycles arrive and later during the same time slot $\lceil\frac{B}{i}\rceil$ packets of $i$  bytes with
$1$ processing cycles arrive. $\OPT$ drops the first subsequence and accepts the second. On the other hand, $\PO$  accepts the first subsequence and drops the second. So during each iteration of the first type $\OPT$ transmits $i(k-1)$  bytes but $\PO$  transmits zero bytes. At the beginning of the next iteration of the first type both algorithms push out already admitted packets that still remain in their buffers.

After the $L$-th iteration, both buffers are nearly full with packets of size $L$, but with $k$ processing cycles in the case of $\PO$ and one residual pass in the case of $\OPT$. Now a sequence of the second type starts. After the last transmission by $\PO$, $k+1$  packets arrive in the following order: first one $L$-byte packet with $k$  passes and thereafter $k$ packets of length $L$ with a single residual pass. The first packet is accepted by $\PO$  and dropped by $\OPT$. The latter $k-1$  packets
are dropped by $\PO$ and accepted by $\OPT$. Each buffer is completely full again. So during each iteration of the second type $\OPT$
transmits $kL$  bytes but $\PO$  only $L$  bytes. After $n$ iterations of the second type the overall transmission of $\OPT$ is $\frac{L(1+L)(k-1)}{2}+knL+B$ while $\PO$ transmits $Ln+B$ bytes. Thus, the lower bound on competitive ratio of $\PO$ is $\frac{L(1+L)(k-1)+2kLn+2B}{2(Ln+B)}$.
\end{proof}

\subsection{Upper bound}

\begin{theorem}
\label{t:lp-po-upper-bound} $\PO$ is at most
$(k+3)$-competitive for LP-based priorities with sufficiently big buffers.
\end{theorem}
\begin{proof}[Sketch]

The mapping routine is unchanged during the transmission phase as in Section~\ref{sec:pq-residual-passes}.
\ignore {
{\bf Mapping routine $\phi$:}
During the transmission phase we distinguish between the three following cases:
\begin{itemize}[noitemsep,topsep=0pt,parsep=0pt,partopsep=0pt,listparindent=0pt,itemindent=0pt]
\item [\bf T0] If neither $\OPT$ and $\PO$ transmit then the mapping remains unchanged.
\item [\bf T1] If $\PO$ transmits a packet $q$ then we remove its mapped image in $\OPT$'s buffer
from future consideration in the mapping.
A subset of these $\OPT$ packets or bytes that stays in the $\OPT$ buffer at the end of transmission are called of type 1.
\item [\bf T2] If $\OPT$ transmits a packet $p$ but its mapped packet $q$ in $\PO$ is not transmitted then $p$ is termed a packet of type 2.
\end{itemize}
}
During the arrival of a packet $p$ at time $t$ Steps A0, A2 and A3 are the same as in Theorem~\ref{t:srpt-po-upper-bound-new}.
At time $t$, denote by $M^O_t$ a set of non-type 1 packets sojourns in $\OPT$ buffer and not mapped by step A2. In addition $M^O_t$ is ordered in non-increasing order of packet length.
All $M^O_t$ packets are grouped into blocks in the following way.
Let $q$ be an $i$-th packet in $\PO$ buffer at time $t$. An $i$-th block is defined in the following way. Consider a minimal set $B_0$ of packets starting from the lowest position that are represented in $M^O_t$ and not covered by any other block whose overall required work is at least $r_t(q)$. If the overall length of all packets in $B_0$ is at least $\ell(q)$ then $B_0$ forms a block. Otherwise, add to $B_0$ a minimal set of packets $B_1$ starting from the first packet that is represented in $M^O_t$ and not covered by $B_0$ such that the overall length of packets in $B_0\cup B_1$ will be at least $\ell(q)$. In this case a set of packets that is covered by $B_0\cup B_1$ defines a block.
Denote by $\ell(X)$ the overall length of packets and by $r_t(X)$ the overall required work in a set of packets $X$ at time $t$. A block $b$ that is mapped to a $\PO$ packet $q$ is called {\em fully mapped} to a packet $q$ at time $t$ if $\ell(b)\geq \ell(q)$ and $r_t(b)\geq r_t(q)$. Observe that it is possible that $\ell(B_0)$ will be less than its $\PO$ counterpart. In this case $\OPT$ may later accept packets that will not be accepted by $\PO$.

We define a new step A1 where the blocks are recomputed after a \emph{$(P,i)$-mapping-shift}.

\begin{itemize}[noitemsep,topsep=0pt,parsep=0pt,partopsep=0pt,listparindent=0pt,itemindent=0pt]
\item [\bf A1] If after $p$ is accepted some $\PO$ packets were dropped then clear the mappings by step A1 between these $\PO$ packets and its $\OPT$ counterparts. If $p$ remains in $\PO$ buffer and $p$ is the $i$-th packet in $\PO$ buffer, perform a \emph{ $(P,i)$-mapping-shift}: clear all mappings by step A1 between packets of the old $l$-th block in $\OPT$ buffer and $l$-th packet in $\PO$ buffer, $l\geq j$, recompute blocks from $j$-th and map packets of $l$-th block to the $l$-th $\PO$ packet if both exist, $l\geq j$.
If $p$ is accepted by $\OPT$ to the $j$-th block, perform  an \emph{$(O,j)$-mapping-shift}: clear all mappings by step A1 between packets in $\OPT$ buffer of the old $l$-th block and $l$-th packet in $\PO$ buffer (if both exist), $l\geq j$, recompute blocks from $j$-th and map packets of $l$-th block to $l$-th $\PO$ packet if both exist, $l\geq j$.
\end{itemize}

Clearly, the mapping is feasible since if
$\OPT$ accepts some packet that is not accepted by $\PO$, $\PO$
buffer contains at least one packet. Since affected
blocks are recomputed after each \emph{ $(P,i)$-mapping-shift } and
\emph{ $(O,j)$- mapping-shift } and by definition of a block $b$
that is mapped to a $\PO$ packet $q$ at time $t$, $\ell(b)\geq
\ell(q)$ and $r_t(b)\geq r_t(q)$. Thus, we will consider
sufficiently big buffers where $\frac{2L-1}{B}$ tends to zero and
because of the above properties of the block step A2 will introduce
at most additional $\epsilon$ value for each $\PO$ packet. Hence,
for each packet $q$ transmitted by $\PO$, $\OPT$ transmits at most
$(k+1)l(q)+\epsilon$ by Steps A1 and A2.
Denote by $T$ the total number of bytes transmitted by $\PO$ and
by $P$ the total number of bytes transmitted by $\OPT$ during
processing of pushed-out $\PO$ packets. Thus, the competitive ratio
is at most $\frac{(k+1+e)T+P}{T}$. Now let us estimate $P$ and
substitute it into the previous expression. For each pushed-out by $\PO$
packet $p$ denote by $T(p)$ a number of time slots when $p$ was HOL
before it was pushed-out.  Clearly, that the process of push-outs
of packets that have positive $T(p)$ will be stopped once all
packets will have a maximal packet length $L$ or it can continue
during each time slot when there is at least one packet in $\PO$
buffer of length smaller than $L$. Moreover, if push-out happens
the buffer occupancy is at least $B-2L+1$. Denote by $\cal P$ a set
of $\PO$ packets pushed-out during this interval of time. So for
each $B-2L+1$ bytes transmitted by $\PO$, $P$ is bounded by
$\sum_{p\in \cal P}{T(p)l(p)}\leq \sum_{p\in \cal P}kl(p)]\leq
kL(L+1)/2$. Thus, $\PO$ is at most
$\frac{2(k+1+\epsilon)B+kL(L+1)}{2(B-2L+1)}$-competitive. For the
buffers that are significantly bigger than $kL(L+1)$, $\PO$ is at
most $k+3$-competitive.
\end{proof}

\section{Buffer Management with MEP-based Priorities}
\label{sec:pq-residual-passes-byte}
In this section we study the performance of a BM implementing PQ, where priorities are set in accordance with the non increasing order of processing cycles divided by packet length. This priority is dubbed the {\em Most Effective Packet} first priority (MEP), or the {\em effective-ratio} priority. Recall that our objective here is to maximize the number of bytes transmitted in total. Again, we concentrate on the push-out case since non-push-out results are similar to Section~\ref{sec:non-preemptive}.

The following theorem provides a lower bound on the performance of the push-out MEP policy.

\begin{theorem}\label{t:mep-po-lower-bound}
$\PO$ with MEP-based priorities is at least $\left(1+\frac{L\ln k - k - L}{B}\right)$-competitive.
\end{theorem}
\begin{proof}
In what follows, we denote $h(n)=\sum_{j=0}^{n}\left\lfloor l/j\right\rfloor$ and assume that $h(k)$ is an integer;
the exact value of $l$ is not important, but the best bound results from $l=L$. Note that $lH(n)-n\le h(n)\le lH(n)$.

In the first burst, we send in an $l$-byte packet with $k$ processing cycles followed by a $\left\lfloor l/k\right\rfloor$-byte packet with
$1$ processing cycle; PO begins processing the first packet while OPT rejects it and processes the second immediately.
On the second time slot, there arrives a $\left\lfloor l/(k-1)\right\rfloor$-byte packet with one processing cycle; PO queues it while
OPT processes it immediately. And so on: on the $i^{\text{th}}$ time slot, a $\left\lfloor l/(k-i)\right\rfloor$-byte packet with one
processing cycle arrives, $i=1,\ldots,k-1$. After $k-1$ steps, OPT has processed $k-1$ packets, has sent out
$h(k)-l$ bytes, and has an empty queue; at the same time, PO has not sent anything out, and its queue contains packets with
one processing cycles and lengths from the set $\left\{\left\lfloor l/(k-i)\right\rfloor, i=0,\ldots,k-1\right\}$
(if $B<h(k)$, PO may push out some of the earlier packets to make room for the later ones).

In the next burst, $B/l$ packets of size $l$ and $1$ processing cycles arrive and flush out everything from PO's queue; we then let
both algorithms finish their packets, sending out $B$ bytes in the process. As a result, PO has sent out $B$ bytes in this process
while OPT has sent $h(k)-l+B\ge l\ln k - k - l + B$ bytes, getting the bound.
\end{proof}

Note that Theorem~\ref{t:mep-po-lower-bound} provides a nontrivial lower bound only for $L > \frac{k}{\ln k-1}$.
A nontrivial upper bound for $\PO$ with MEP priorities remains an interesting open problem.

\section{Simulation study}\label{sec:simulations}

\subsection{General remarks}

In order to obtain a better understanding of the differences between our proposed solutions, we conducted a simulation study where we evaluate the performance of each policy in terms of throughput and address the effect of variable processing requirements on the average delay in the system.

Publicly available traffic traces (such as CAIDA~\cite{CAIDA}) do not contain, to the best of our knowledge, information on the processing requirements of packets. Furthermore, these requirements are difficult to extract since they depend on the specific hardware and NP configuration of the network elements. Another handicap of such traces is that they provide no information about time-scale, and specifically, how long should a time-slot last. This information is essential in our model in order to determine both the number of processing cycles per time-slot, as well as traffic burstiness.
We therefore perform our simulations on synthetic traces.

Our simulation results are based on traffic composed of the interleaving of 100 independent sources, with each source generated by an on-off bursty process modeled by a Markov-modulated Poisson process (MMPP).
During every time slot, each source has probability $0.05$ to be switched on, and once switched on, probability $0.2$ to be switched back off.
When a source is on, it emits packets with intensity $\lambda_{\mathrm{on}}$, which represents one of the parameters governing traffic generation. Each generated packet is assigned two parameters:
\begin{inparaenum}[(i)]
\item required processing chosen uniformly at random from $\set{1,\ldots,k}$ ($k$ being the maximum amount of processing required by any packet), and
\item packet length, chosen uniformly at random from $\set{1,\ldots,L}$ ($L$ being the maximum length of a packet in the system).
\end{inparaenum}
Each of our results follows from simulating the system for $5{,}000{,}000$ time slots; we allowed different parameters to vary in each set of simulations in order to better understand the effect each parameter has on system performance and further validate our analytic results and algorithmic insights. 

We simulated the throughput performance of our three proposed policies, all based on the greedy algorithm depicted in Algorithm~\ref{alg:pq-1}:
\begin{inparaenum}[(i)]
\item SRPT,
\item LP, and
\item MEP.
\end{inparaenum}
In order to obtain a better qualitative differentiation between the policies, we compared their throughput performance with that of a ``virtual'' policy, which serves as an approximate upper bound on the optimal throughput possible. This virtual policy essentially transforms each arriving packet requiring $k' \leq k$ processing cycles, and having length $\ell' \leq K$, into $k'$ distinct packets, each requiring one processing cycle, and having length $\ell'/k'$, using the LP/MEP as the scheduling and admission criteria (they are equivalent for such virtual inputs). Clearly the performance of this virtual policy serves as an approximate upper bound on the performance of the optimal policy, since this policy
profits from any partial processing of a packet.
We use this approximation since finding the actual optimal algorithm would be computationally prohibitive:
even identifying the best set of packets to store in a single time step is equivalent to the knapsack problem which is NP-hard.
In Figures~\ref{fig:simulations_lambda}(a), \ref{fig:simulations_k}(a), and \ref{fig:simulations_b}(a), which demonstrate the throughput performance of the system, the y-axis represents the ratio between the throughput obtained by a policy and the throughput obtained by the virtual policy (which serves as an approximate upper bound on the optimal policy).

Another set of results produced by our simulation study deals with the average queuing delay of packets for each of the policies considered. As mentioned in the introduction, queueing delay has long been known to be directly related to the buffer size available for the queue. Our work tries to shed light on the role of variable processing requirements as a major factor affecting queueing delay in such heterogeneous environments and relate this latency performance to that of the attainable throughput.
In Figures~\ref{fig:simulations_lambda}(b), \ref{fig:simulations_k}(b), and \ref{fig:simulations_b}(b), which demonstrate the average latency in the system, the y-axis represents the average latency (in time slots) over all packets delivered.

We present here only a small sample of our results, aiming to explore the effect of various parameters examined in our study. Specifically, we consider the effect of offered-load, average number of processing required by a packet, buffer size, and average packet length.
For each of the first three parameters here mentioned, we present a cross section of the effect of average packet length by providing three plots corresponding to maximum allowed packet length values $L=10,15,25$.

\subsection{Varying traffic intensity}
\label{sec:sim:load}

\begin{figure*}
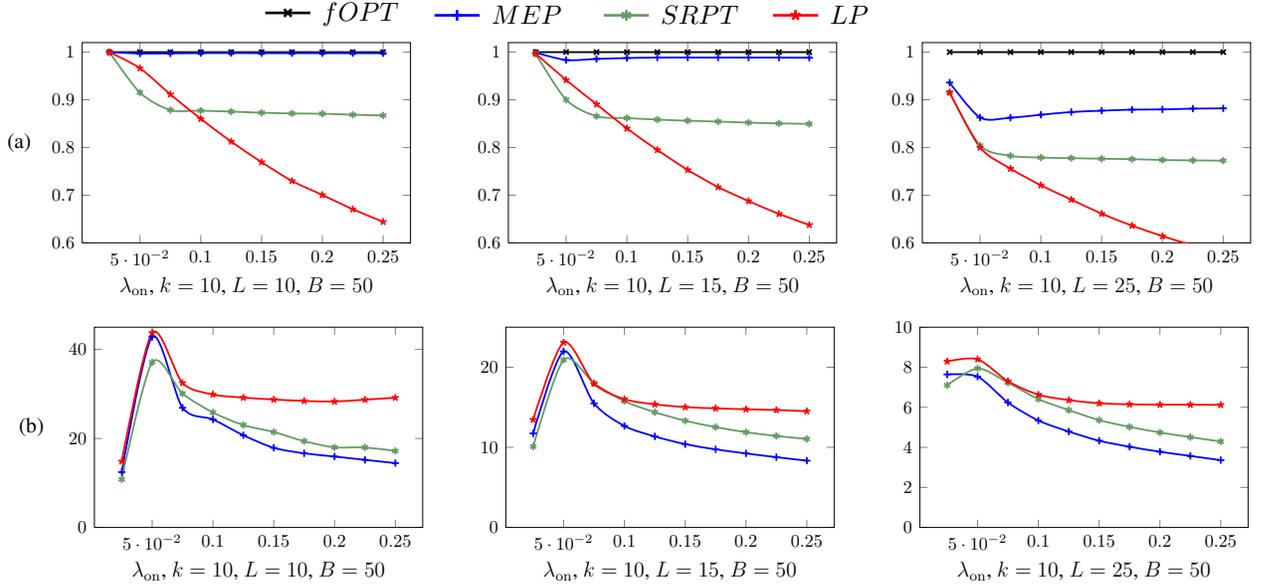

\begin{center}
\simullegend

\begin{tabular}{ccc}
\variablelambda{10}{}{1.02}{0.6}{\rotatebox{270}{(a)}} & \variablelambda{15}{}{1.02}{0.6}{ } & \variablelambda{25}{}{1.02}{0.6}{} \\\vspace{\alittle}
\variablelambda{10}{.lat}{45}{0}{\rotatebox{270}{(b)}} & \variablelambda{15}{.lat}{25}{0}{ } & \variablelambda{25}{.lat}{10}{0}{} \\\vspace{\alittle}
\end{tabular}
\end{center}

\caption{Throughput performance (a) and latency (b) as a function of incoming stream intensity $\lambda_{\mathrm{on}}$ for three different values of maximal packet length $L$.
}\label{fig:simulations_lambda}
\end{figure*}

Figure~\ref{fig:simulations_lambda} shows the system performance as a function of increased average load, where we increase the rate of each independent source by increasing the parameter $\lambda_{\mathrm{on}}$ which governs packet intensity during a burst period.
Figure~\ref{fig:simulations_lambda}(a) shows that, in general, MEP is the best policy (this will always be the case throughout our
simulations). However, when examining the other two policies, although as traffic intensity increases SRPT significantly outperforms LP, under low load conditions and small values of $L$, LP outperforms SRPT; This indicates that  under moderate load conditions, and when packet length variability is small, it is best to prioritize longer packets rather than by their processing requirements.
When either as traffic intensity increases (or packet length variability grows), the system will be prone to increased congestion, whose alleviation is possibly by preferring packets which take a shorter time to process.
As for the latency, Figure~\ref{fig:simulations_lambda}(b) shows that average latency increases up to a certain point, and then steadily decreases. This increase occurs in moderate load conditions, where all algorithms are ``forced'' to accept non-favorable packets. However, as traffic intensity increases, all algorithms have a better selection of packet to accept, and each will focus on its more preferable packets, thus resulting in decreasing packet latency.

\subsection{Increasingly heterogeneous processing requirements}

\begin{figure*}
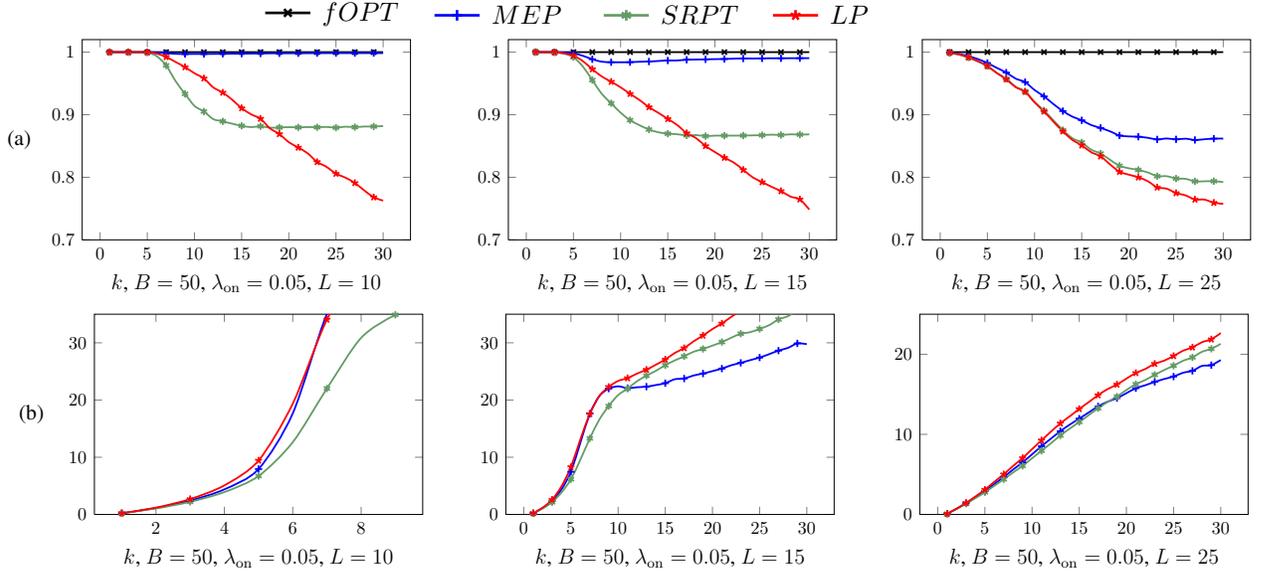

\begin{center}
\simullegend

\begin{tabular}{ccc}
\variablek{10}{}{1.02}{0.7}{\rotatebox{270}{(a)}} & \variablek{15}{}{1.02}{0.7}{ } & \variablek{25}{}{1.02}{0.7}{ } \\\vspace{\alittle}
\variablek{10}{.lat}{35}{0}{\rotatebox{270}{(b)}} & \variablek{15}{.lat}{35}{0}{ } & \variablek{25}{.lat}{25}{0}{ } \\\vspace{\alittle}
\end{tabular}
\end{center}

\caption{Throughput performance (a) and latency (b) as a function of maximal required work $k$ for three different values of maximal packet length $L$.
}\label{fig:simulations_k}
\end{figure*}

Figure~\ref{fig:simulations_k} shows the system performance as we allow packet processing requirement to increase, both in value and in variability.
As demonstrated in Figure~\ref{fig:simulations_k}(a), while the MEP policy outperforms both other policies, for relatively small $L$ we observe a transition from LP to SRPT as the second best policy. One can see that while the average number of processing cycles is relatively low, the LP policy outperforms the SRPT policy, while as the average number of required processing increases beyond some threshold, SRPT becomes superior to LP. This behavior is similar to that observed in the study of the effect of traffic-intensity on the performance in Section~\ref{sec:sim:load}. This coincides with the intuition that the actual notion of load in the system is actually the product of the average required processing and packet arrival rate.
The simulation results presenting the effect of increasing load and increasing required processing on the system's throughput are in accordance with the results obtained in our analytic study, which show that the ratio between the parameters $k$ and $L$ indeed corresponds to which of the policies is expected to be superior.
The latency graphs here (Figure~\ref{fig:simulations_k}(b)) are mostly strictly increasing, with latency becoming pronounced as the overall arrival load (in the sense just described) topping the system's processing service rate (this occurs at around $k=5$). When examining the differences in latency as average packet length increases, once can see that the average latency (for the same values of $k$) is inversely proportional to the average packet length. This is due to the fact that every successful transmission when packets are larger leaves reduces the delay of packets remaining in the queue.

\subsection{The effect of buffer size}

\begin{figure*}
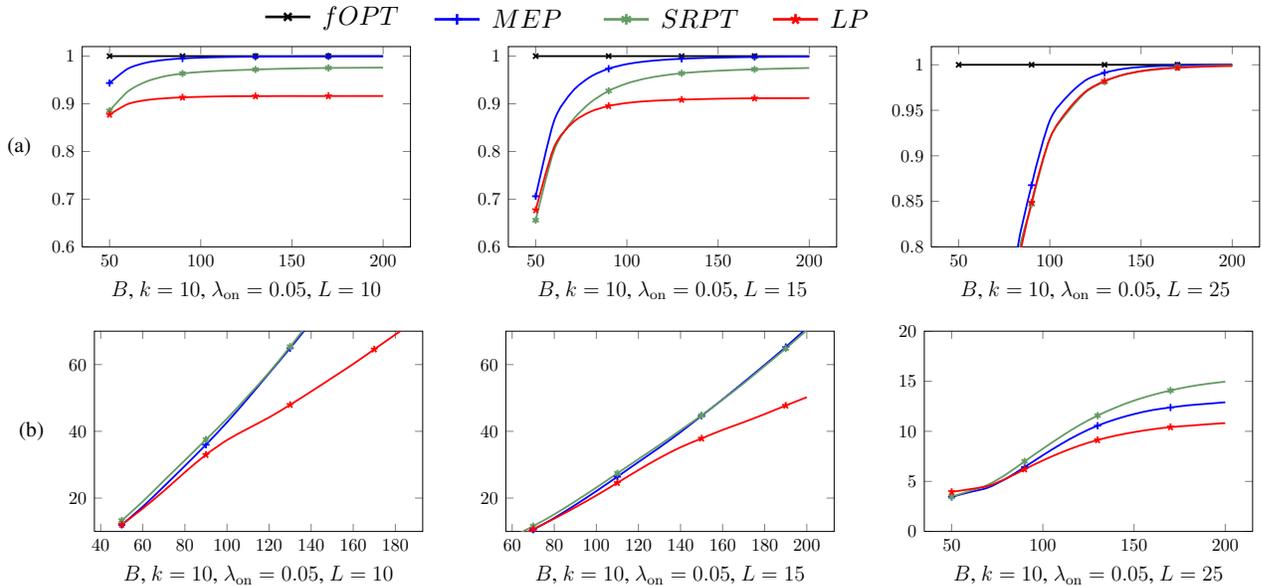

\begin{center}
\simullegend

\begin{tabular}{ccc}
\variableb{10}{}{1.02}{0.6}{\rotatebox{270}{(a)}} & \variableb{15}{}{1.02}{0.6}{ } & \variableb{25}{}{1.02}{0.8}{ } \\\vspace{\alittle}
\variableb{10}{.lat}{70}{10}{\rotatebox{270}{(b)}} & \variableb{15}{.lat}{70}{10}{ } & \variableb{25}{.lat}{20}{0}{ } \\\vspace{\alittle}
\end{tabular}
\end{center}\vspace{-.5cm}
\caption{Throughput performance (a) and latency (b) as a function of buffer size $B$ for three different values of maximal packet length $L$.
}\label{fig:simulations_b}
\end{figure*}

Figure~\ref{fig:simulations_b} shows the effect buffer size has on system performance.
In Figure~\ref{fig:simulations_b}(a) One can see that the buffer size has relatively little effect on the differences between the policies in terms of throughput: all three policies relatively quickly achieve their corresponding maximal performance
and stay there as buffer size grows further; This is due to the fact that an beyond a certain point, packet arrival rate is smoothed by the availability of buffer space.
In terms of latency, Figure~\ref{fig:simulations_b}(b) shows a steady increase in latency, which should be ascribed to queueing delay. However, the LP policy exhibits the best performance in these scenarios since it favors the transmission of longer packets first, which alleviate the latency sensed by the remaining packets in the buffer.

In general, our results clearly show that the MEP policy is better than both other policies with respect to throughput.
Note that in terms of latency the best policy (MEP) does not necessarily outperform other policies: since it processes
more packets, some of them must wait for their turn longer.

Our simulation results and the insights they provide serve as a rule-of-thumb in choosing the best policy for a specific network scenario, depending on the expected traffic characteristics.

\section{Conclusion}\label{sec:conclusion}
Increasingly heterogeneous packet processing requirements in modern networks pose novel design
challenges to NP architects. In this work we study the impact of two important
characteristics, maximal required processing $k$ and maximal packet size $L$, and show the significance of the relationship between $k$ and $L$. We introduce three different priority regimes
for processing: SRPT, LP, and MEP, and study their performance in queues with bounded buffers. We present results for both non-push-out, as well as push-out buffer management algorithms, which give guarantees on the worst-case performance of such algorithms, without resorting to any assumptions on the process generating the traffic. Due to this approach, are results can be globally applicable, in various networking environments which may deal with highly heterogenous traffic patterns.

Our results show that implementing a push-out mechanism, although potentially costly in terms of vendor implementation, has a significant impact on the system's performance, primarily in terms of throughput.
Remaining open questions include closing the gaps between the upper and lower bounds (shown in Table~\ref{tbl}), and predominantly determining the performance of the $\PO$ algorithm with MEP priorities.



\bibliographystyle{plain}
\bibliography{np}

\end{document}